\documentclass[sigconf]{acmart}

\usepackage{subfig, float}
\usepackage{amssymb}
\usepackage{amsmath}
\usepackage{graphicx}
\usepackage{array}
\usepackage{colortbl}
\usepackage{makecell}
\usepackage{bm}

\usepackage[section]{algorithm}
\usepackage{algpseudocode}

\newtheorem{theorem}{Theorem}

\usepackage[inline]{trackchanges}
\addeditor{DJORDJE}

\AtBeginDocument{%
  \providecommand\BibTeX{{%
    \normalfont B\kern-0.5em{\scshape i\kern-0.25em b}\kern-0.8em\TeX}}}

\setcopyright{acmcopyright}
\copyrightyear{2020}
\acmYear{2020}
\acmDOI{}

\acmConference[AdKDD~'20]{San Diego '20: SIG Conference on Knowledge Discovery and Data Mining}{August 2020}{San Diego, California, USA}
\acmBooktitle{AdKDD~'20: SIG Conference on Knowledge Discovery and Data Mining, August 2020, San Diego, California, USA}

\begin{document}




\title{Bid Shading by Win-Rate Estimation and Surplus Maximization}

    

\author{Shengjun Pan, Brendan Kitts, Tian Zhou, Hao He, Bharatbhushan Shetty, Aaron Flores, Djordje Gligorijevic, Junwei Pan, Tingyu Mao, San Gultekin and Jianlong Zhang}

\affiliation{
  \institution{Verizon Media}
  \city{Sunnyvale}
  \state{California}
}

\email{{alanpan, brendan.kitts, tian.zhou, hao.he, bharatbs, aaron.flores}@verizonmedia.com}
\email{{djordje, jwpan, tingyu.mao, sgultekin, jianlong}@verizonmedia.com}

\renewcommand{\shortauthors}{S Pan, B Kitts, T Zhou, H He, B Shetty, A Flores, D Gligorijevic, J Pan, T Mao, S Gultekin and J Zhang}

\begin{abstract}
This paper describes a new win-rate based bid shading algorithm (WR) that does not rely on the minimum-bid-to-win feedback from a Sell-Side Platform (SSP). The method uses a modified logistic regression to predict the profit from each possible shaded bid price. The function form allows fast maximization at run-time, a key requirement for Real-Time Bidding (RTB) systems. We report production results from this method along with several other algorithms.
We found that bid shading, in general, can deliver significant value to advertisers, reducing price per impression to about 55\% of the unshaded cost. Further, the particular approach described in this paper captures 7\% more profit for advertisers, than do benchmark methods of just bidding the most probable winning price. We also report 4.3\% higher surplus than an industry Sell-Side Platform shading service. Furthermore, we observed 3\%~--~7\% lower eCPM, eCPC and eCPA when the algorithm was integrated with budget controllers. We attribute the gains above as being mainly due to the explicit maximization of the surplus function, and note that other algorithms can take advantage of this same approach.
\end{abstract}





\maketitle

\section{Introduction}
\label{sec:introduction}
Online Advertising auctions have been dominated by Second Priced Auctions (SPAs) since their early implementations in the 1990s. Google famously used Second Price Auctions for its Adwords and Adsense auctions, and, in 2017, generated 90\% of its revenue from Second Price Auctions \cite{google:sec}. 
However, there was a dramatic shift in online advertising between 2018 and 2019. 
As of 2020, almost all major display ad auctions have switched from Second to First Price Auctions (FPAs) \cite{google:rtb,google:rtb2}. Several factors conspired to drive the industry towards the adoption of FPA, including the widespread growth of header bidding with its incompatibility with SPAs ~\cite{hearts2018:sold},  increased demand for transparency and accountability \cite{chari1992us,sluis:guardiansuesrubicon,getintent:rtb1,rubicon:openletter}, and yield concerns \cite{appnexus:BPO, rubicon:EMR, kitts:fpa}.

Unfortunately for advertisers, First Price Auctions leave private value bidders susceptible to over-paying. For instance, if the bidder's private value of an impression was \$10.00, and the winner knew the second placed bidder's price was just \$1.00, they could bid just \$1.01 and effectively collect a \$8.99 profit. If they instead bid their private value, they would be charged the entirety of the \$10.00 and they would have \$0  profit!

The practice of strategically decreasing bid price below the buyer's private value is known as \textit{bid shading}. Bid shading has been observed in a variety of real world auctions including FCC Spectrum \cite{chakravorti1995auctioning}, US Oil Deposits \cite{capen:auctions}, Cattle auctions \cite{crespi2005multinomial}, US Treasury auctions \cite{hortaccsu2018bid} and others. 
Despite its widespread use, there has been little work done on methods to systematically exploit shading, particularly when data is available to make it possible to predict auction clearing prices.

\section{The Bid Shading Problem}
\label{sec:bid-shading}
Given bid request $b_i$, and a valuation $V_i$, if we won the impression, which represents how much the advertiser expects to capture from the impression, how much should the advertiser discount their valuation? Assuming that the valuation $V_i$ is an accurate representation of the dollar value that the advertiser expects to obtain, and the bid $b_i=g_iV_i$ is also in real dollars, the advertiser's financial gain, or \emph{surplus}, is equal to: 
\begin{align}
surplus & \stackrel{\rm def}{=}
\sum_i(V_i - g_iV_i)\, \mathbf{I}(g_iV_i> \hat b_i),
\label{eq:surplus_per_response}
\end{align}
where $g_i\in(0, 1]$ is the shading factor to apply to the bidder's private value $V_i$, $\hat b_i$ is the minimum bid price to win,
and $\mathbf{I} (b_i> \hat b_i) = 1$ if the impression is won, and $0$ otherwise.
The task is to find a shading factor $g_i$ that maximizes the surplus to the advertiser.

\section{Previous Work}
\label{sec:related-work}

\subsection{Bid Shading Theory}

Bid shading is a common tactic in repeated First Price Auctions.
\cite{zulehner2009bidding} found robust evidence of shading in Austrian livestock auctions , \cite{crespi2005multinomial} reported shading in a Texas cattle market, and \cite{hortaccsu2018bid} found the practice in auctions for US Treasury notes.

Auctions generally need to be repeated and predictable for bid shading to be practically feasible, but under these conditions, it often occurs organically. Pownall and Wolk (2013) showed that bid shading for repeated internet auction prices increased over time; by about 26\% after 10 iterations \cite{pownall2013bidding}. When there are enough repeated games bidders can even develop collusive shading strategies where bidders actively coordinate to have low bids \cite{lengwiler2010auctions, hendricks1989collusion}. 

Although behavior varies from auction to auction, several studies have shown that the magnitude of shading tends to increase with the average price on the auction  \cite{chakravorti1995auctioning, battigalli2003rationalizable, hortaccsu2018bid}. This is likely to occur because of the more substantial losses involved on higher priced auctions, if shading isn't sufficient. 
This result suggests that using a measure of the expense of the auction is valuable when trying to estimate the shading factor - a finding we revisit later in Section \ref{sec:insights}.

In situations where the supply is plentiful, and demand limited, buyers can shade deeper. In looking at this phenomemon in the US Treasury Market, Hortacsu et. al. (2017) find that large institutional buyers on average shade more aggressively than small indirect buyers \cite{hortaccsu2018bid} . This seems to be because these large buyers effectively control a large percentage of bidders, and so it is almost like they are able to coordinate the buying of multiple buyers. They can therefore drive the bid prices for a large percentage of bidders down, whilst still meeting their goals. 

\subsection{Previous Algorithms}

In 2018 and 2019, Rubicon~\cite{adx:rubicon, rubicon:EMR}, AppNexus~\cite{appnexus:BPO} and Google~\cite{shields:googlefp, sluis:googlefp, google:rtb, google:rtbprotocol} all released Sell-Side bid shading services. Leading up to this, there had been reports of dramatically lower ROI from the new First Price Auctions~\cite{hearts2018:sold, kitts:fpa}. Never-the-less, this is a surprising move as Sell-Side Platforms are potentially decreasing their yield, and they clearly have a different incentive from buyers. The sell-side algorithms seem to reflect this incentive difference. 
The descriptions of these services suggest that they try to keep bid prices high enough to maintain a set win-rate, but preventing the bid price from becoming too extreme; which might risk an advertiser to halt their bidding due to poor Return on Investment. Rubicon released data suggesting that their service decreases First Price CPMs by a modest 5\% over 4 months~\cite{rubicon:EMR}. AppNexus reported that prices under their service were 25\% lower over 100 days~\cite{appnexus:BPO}. We tried one of the services, and recorded the shading distribution in Figure~\ref{fig:winrate_threedistributions}.
Most of the bid shades were about 90\%, which is conservative for our problem. Further analysis on Sell-Side "Bid Shaders" are in Section~\ref{sec:experiments}.

On the Demand Side, a variety of algorithms have been explored, although generally not exactly for bid shading applications.
\cite{wu2015predicting} developed a censored winning bid probability estimator. They observed that when a bidder submitted a bid and lost, the information gained is that the winning price is somewhere above the submitted price, and when a bidder submits a bid and wins, the minimum bid to win is at a price somewhere below their submitted bid. Using these two cases, the authors developed a Maximum Likelihood procedure to estimate the probability of the winning bid being any of the bid prices. This created a distribution of the probable winning bids, with the most likely winning price being used for bidding.  \cite{wu2018deep} extended their work to using a neural network to estimate the parameters of the win probability distribution.

An unpublished implementation \cite{dsp:lr} used Logistic Regression to predict the optimal bid shading factor using features in the request. The predicted factor was then used as a multiplier on the unshaded bid price.

The approaches described above \cite{wu2015predicting, wu2018deep, dsp:lr} all focus on predicting the probable winning bid price.
However, the surplus maximum is very different from the minimum bid to win. 
An accurate (unbiased, symmetric noise) win probability estimator will be below the winning bid price about 50\% of the time - this means that 50\% of the surplus won't be captured \textit{by design}. If the change in new impressions captured at a higher bid price, over-weights the marginal decrease in profitability per impression, the optimum for surplus can be higher than the most probable bid.

Unpublished work \cite{niklas:shading} is one of the few that we know of to attempt to explicitly maximize the surplus function. These authors estimate shading factors for a set of fixed segments based on three bid samples taken in real-time to estimate the local surplus landscape. However the approach has many drawbacks: the segments have to be predetermined and finding a suitable segment definition requires substantial analysis. 
The information across segments is not shared, which is a problem for segments that do not have enough traffic. Further, the set of possible segments quickly explode as the number of variables used to define them increases.
The approach taken in this paper uses a model to estimate the surplus function, and so a very large number of features can be used, and model induction is also automated, easy to maintain, and improve. 

In order to compare the method we used to prior work, we have included an implementation of the Logistic Regression algorithm from 
\cite{dsp:lr}, the Distribution Estimator algorithm from  \cite{wu2015predicting}, and the Segment-based Surplus maximizer \cite{niklas:shading} in the benchmarks which we use to analyze algorithm performance in Section \ref{sec:experiments}.

\section{Canonical Algorithm}
\label{sec:methodology}
Given a bid request for First Price Auction, 
let $x_1, x_2, \ldots, x_k$ be the set of publisher and user attributes that we will use to find the best bid price $b^*$. Let $\hat b$ be the highest bid price from other competing bidders, which value is unknown. Note that $\hat b$ depends on both attributes $x_i$s which represent the item that is being auctioned, and  external competing bidder behavior. 
$\hat b$ follows an unknown distribution $\mathcal{D}_{\hat b\mid x_1, x_2, \ldots, x_k}$ with cumulative probability distribution $\text{cdf}_{\hat b\mid x_1, x_2, \ldots, x_k}$.
When the context is clear, we use $\mathcal{D}_{\hat b}$ and $\text{cdf}_{\hat b}$ for simplicity.

If the distribution $\mathcal{D}_{\hat b}$ is known, we can calculate the optimal bid price $b^*$ directly as follows. Let $\mathbf{I}(b>\hat b)$ be 1 if $b>\hat b$ and 0 otherwise, which indicates if the submitted price $b$ wins the auction. Then the \text{surplus} when the submitted price is $b$ would be
\begin{align}
    surplus &= (V-b)\mathbf{I}(b>\hat b)
    =\begin{cases}
    V-b, & \text{if }b > \hat b, \\
    0, & \text{otherwise}.
    \end{cases}
\end{align}
The optimal bid price can be calculated as the price that maximizes the expected surplus

\begin{align}
    b^* &= \mathop{\arg\max}_{b>0}\mathbb{E}[surplus] \nonumber\\
    & = \mathop{\arg\max}_{b>0}\mathbb{E}\left[(V-b)\,\mathbf{I}(b>\hat b)\right] \nonumber\\
    & =\mathop{\arg\max}_{b>0}(V-b)\, \text{cdf}_{\hat b}(b).\label{eq:max-surplus}
\end{align}
For simple forms of $\text{cdf}_{\hat b}(b)$, the optimization problem~\eqref{eq:max-surplus} can be solved analytically. For example, suppose $\hat b$ distributes uniformly over the interval $[B_0, B_1]$, where $0\leq B_0 < B_1$. This produces a  $\text{cdf}_{\hat b}(b)$ that is piece-wize linear, with a flat region of 0.0 from $[0, B_0]$, a constant slope from $[B_0, B_1]$, and another flat region of 1.0 above $B_1$. The bid price $b^*$ that maximizes the surplus can be calculated as below 
\begin{align*}
\mathbb{E}[surplus] &= (V-b)\,\text{cdf}_{\hat b}(b) \\
& = \begin{cases}
    0, & \text{ if } b < B_0, \\
    (V-b)(b-B_0)/(B_1-B_0), &\text{ if } B_0 \leq b \leq B_1, \\
    V-b, & \text{ if } b > B_1.
\end{cases}
\end{align*}
It is straightforward to see that
\[\max\mathbb{E}[surplus] =\begin{cases}
    \frac{(V-B_0)^2}{4(B_1-B_0)}\text{ at } b^*=\frac{V-B_0}{2}, & \text{ if } V\leq 2B_1-B_0, \\
    V-B_1 \text{ at } b^*=B_1, & \text { if } V>2B_1-B_0.
\end{cases}
\]

\begin{figure}
\centering

\subfloat{\includegraphics[trim=1cm 7.25cm 1cm 7.25cm, clip=true, width=0.5\columnwidth]{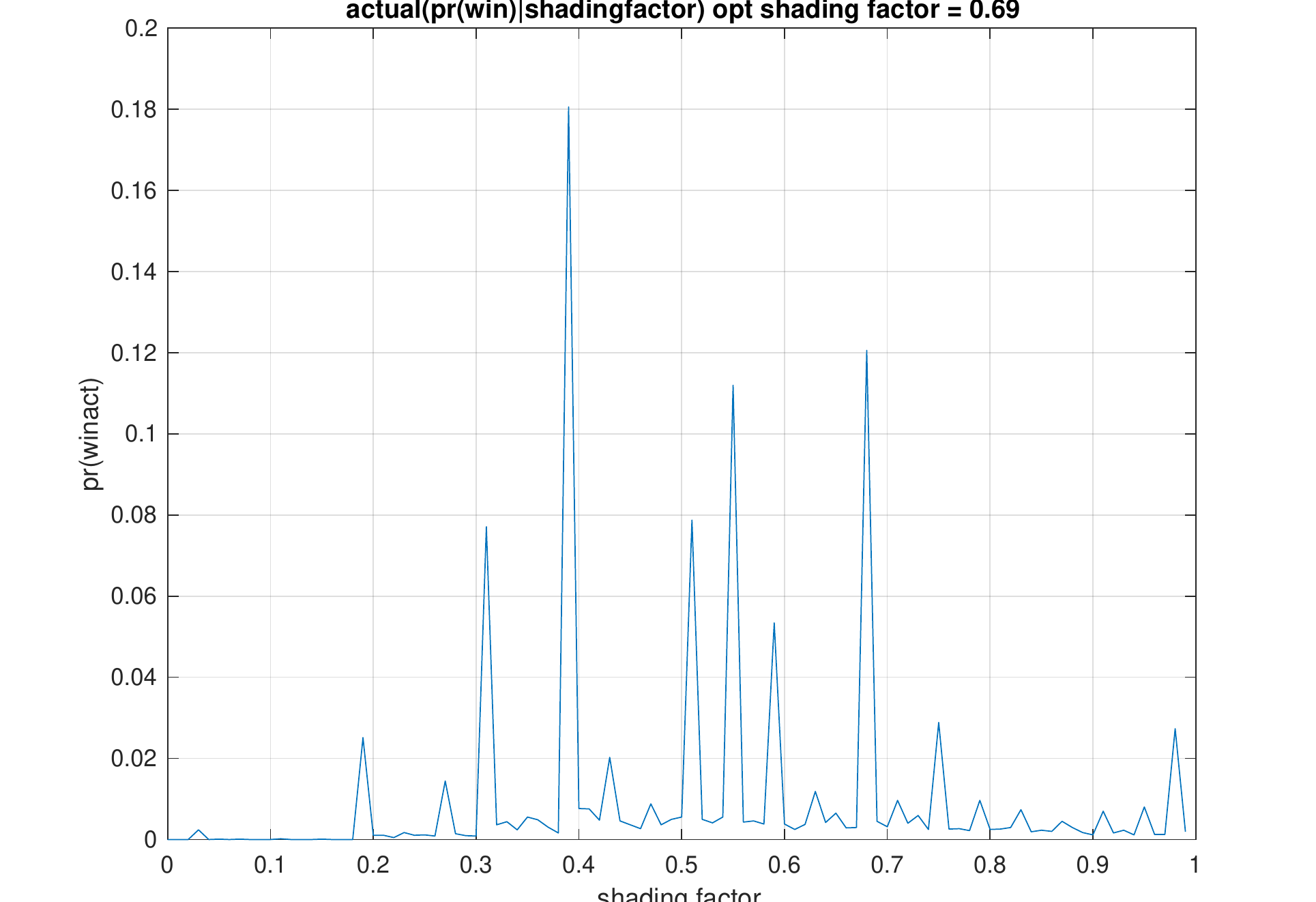}} 
\subfloat{\includegraphics[trim=1cm 7.25cm 1cm 7.25cm, clip=true, width=0.5\columnwidth]{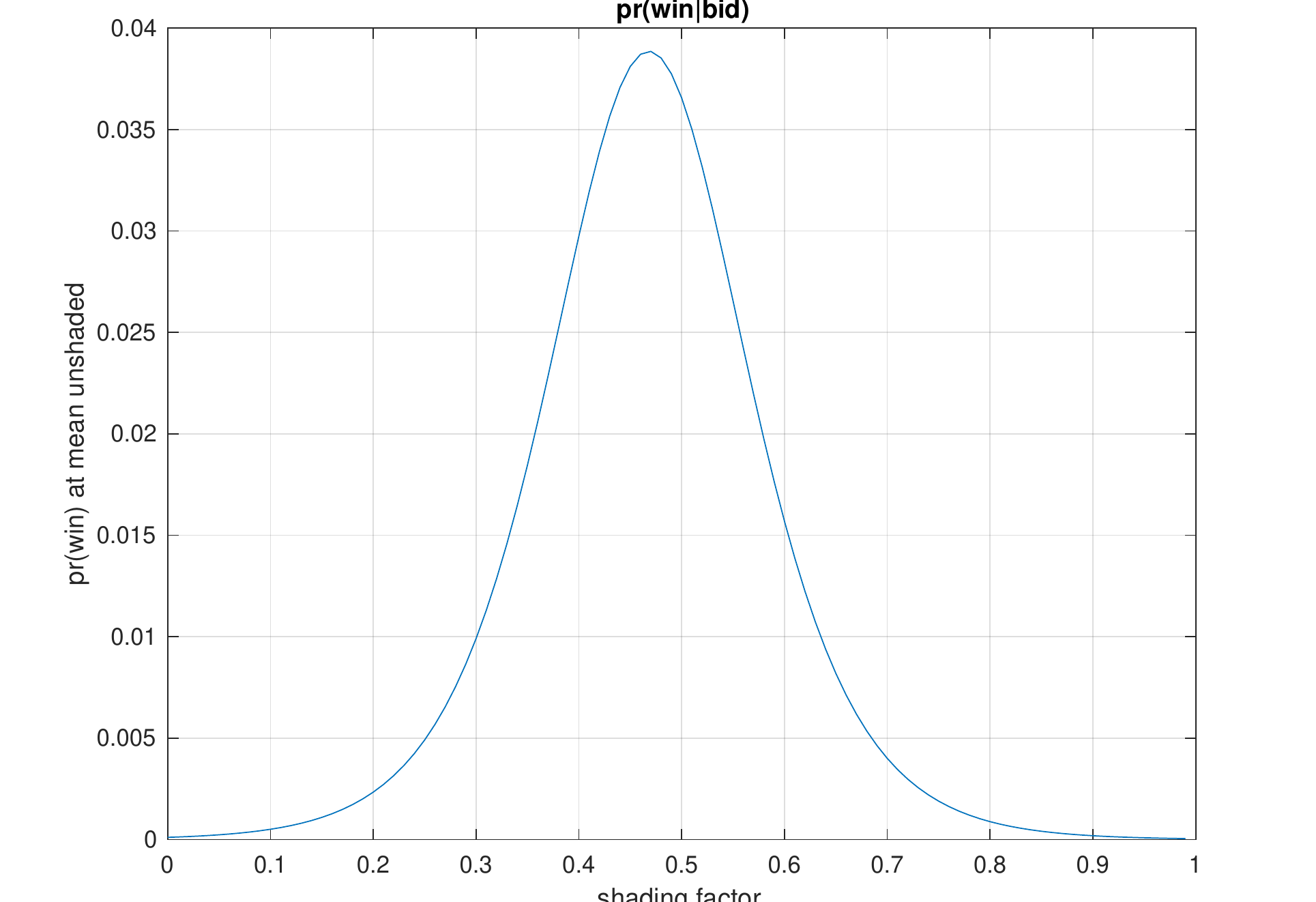}}\\
\subfloat{\includegraphics[trim=1cm 7.25cm 1cm 7.25cm, clip=true, width=0.5\columnwidth]{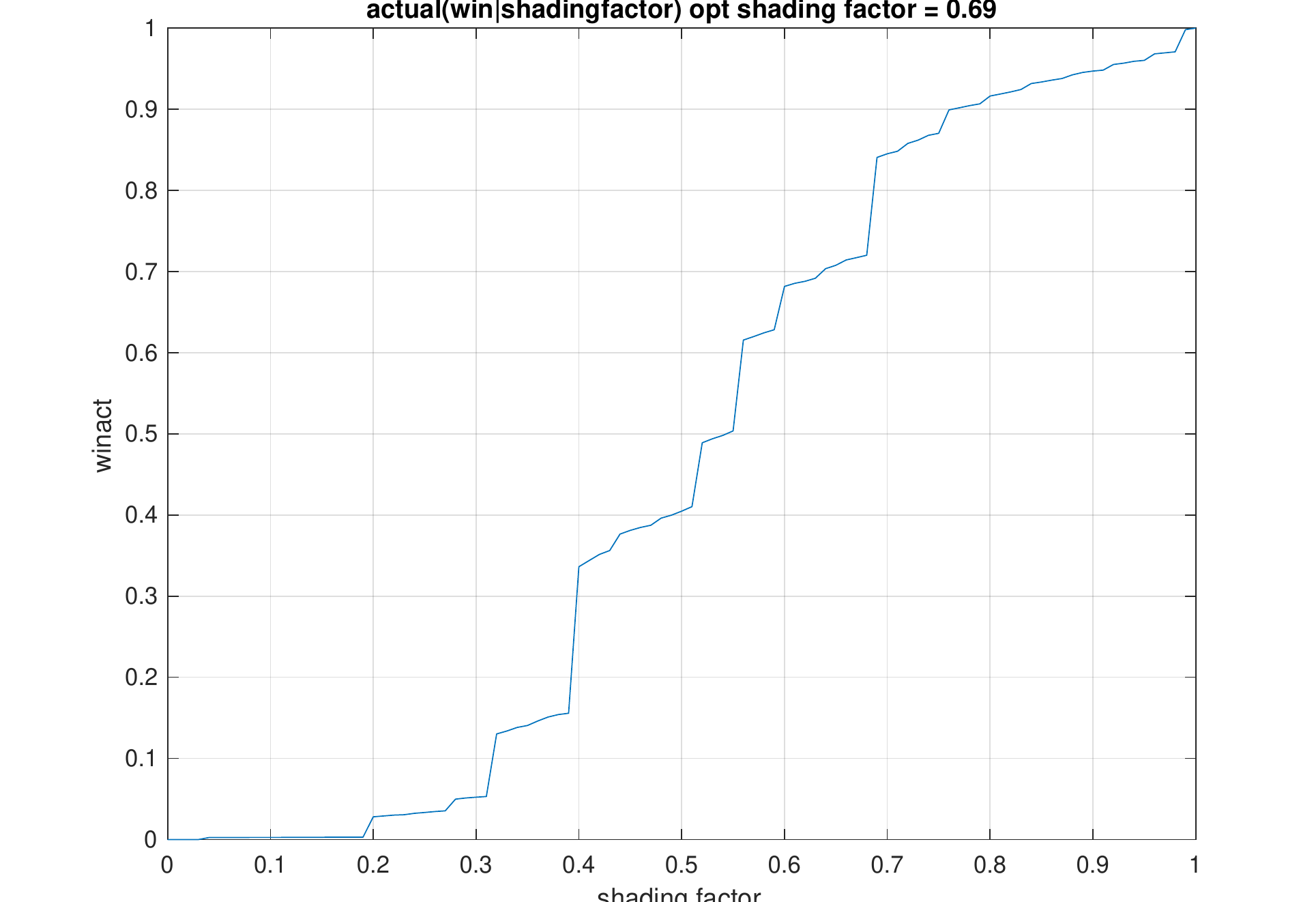}} 
\subfloat{\includegraphics[trim=1cm 7.25cm 1cm 7.25cm, clip=true, width=0.5\columnwidth]{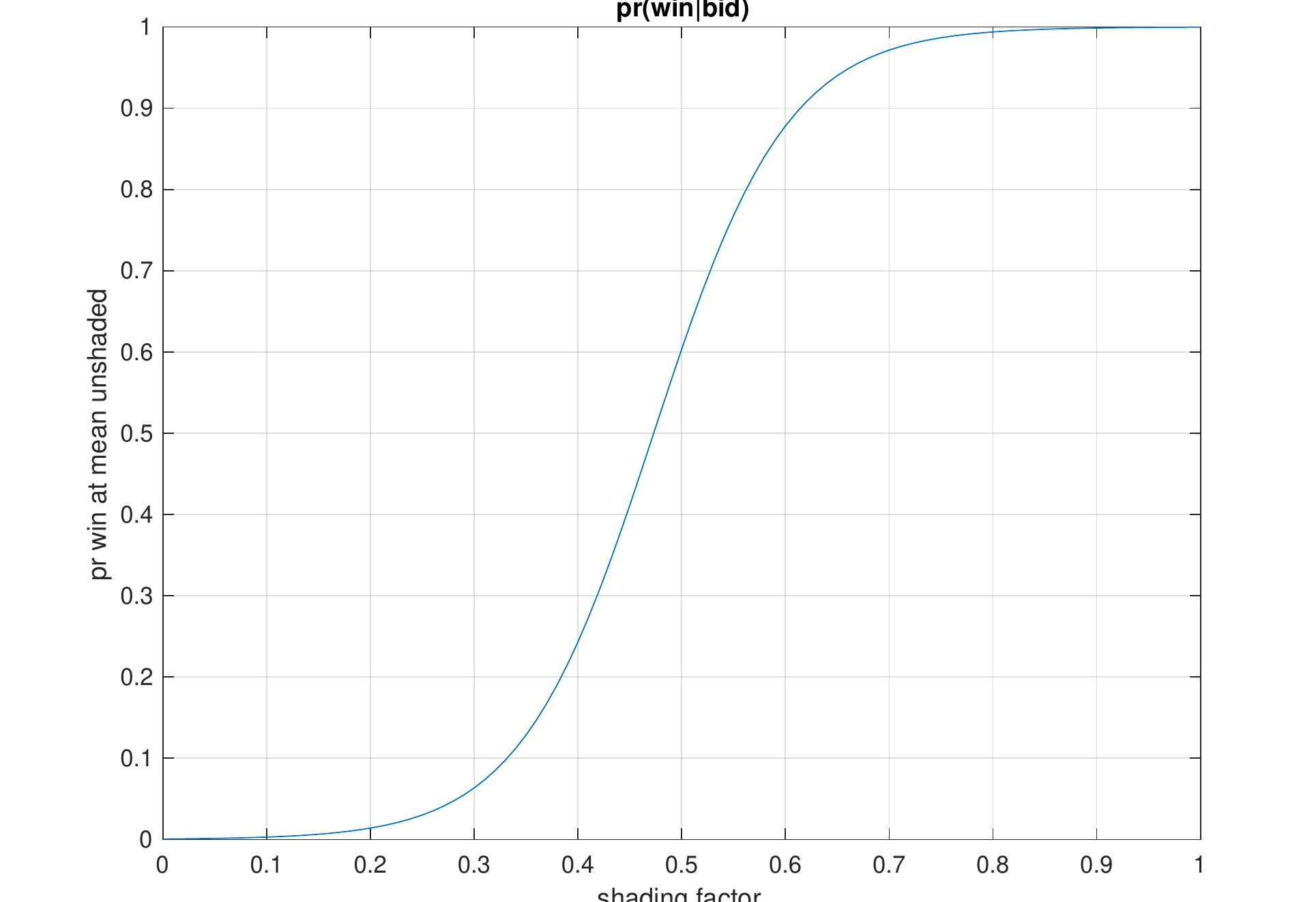}}\\
\subfloat{\includegraphics[trim=1cm 7.25cm 1cm 7.25cm, clip=true, width=0.5\columnwidth]{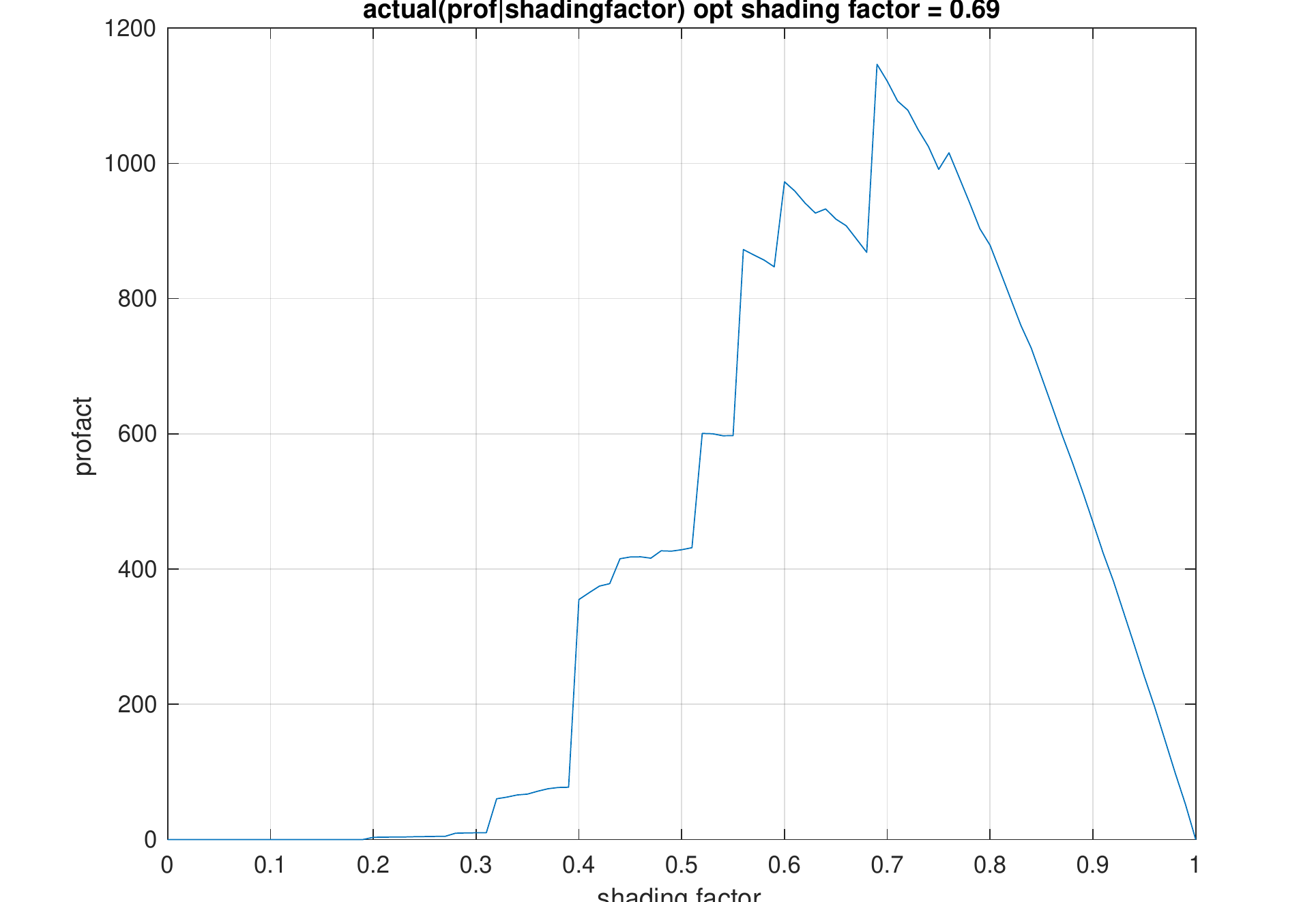}} 
\subfloat{\includegraphics[trim=1cm 7.25cm 1cm 7.25cm, clip=true, width=0.5\columnwidth]{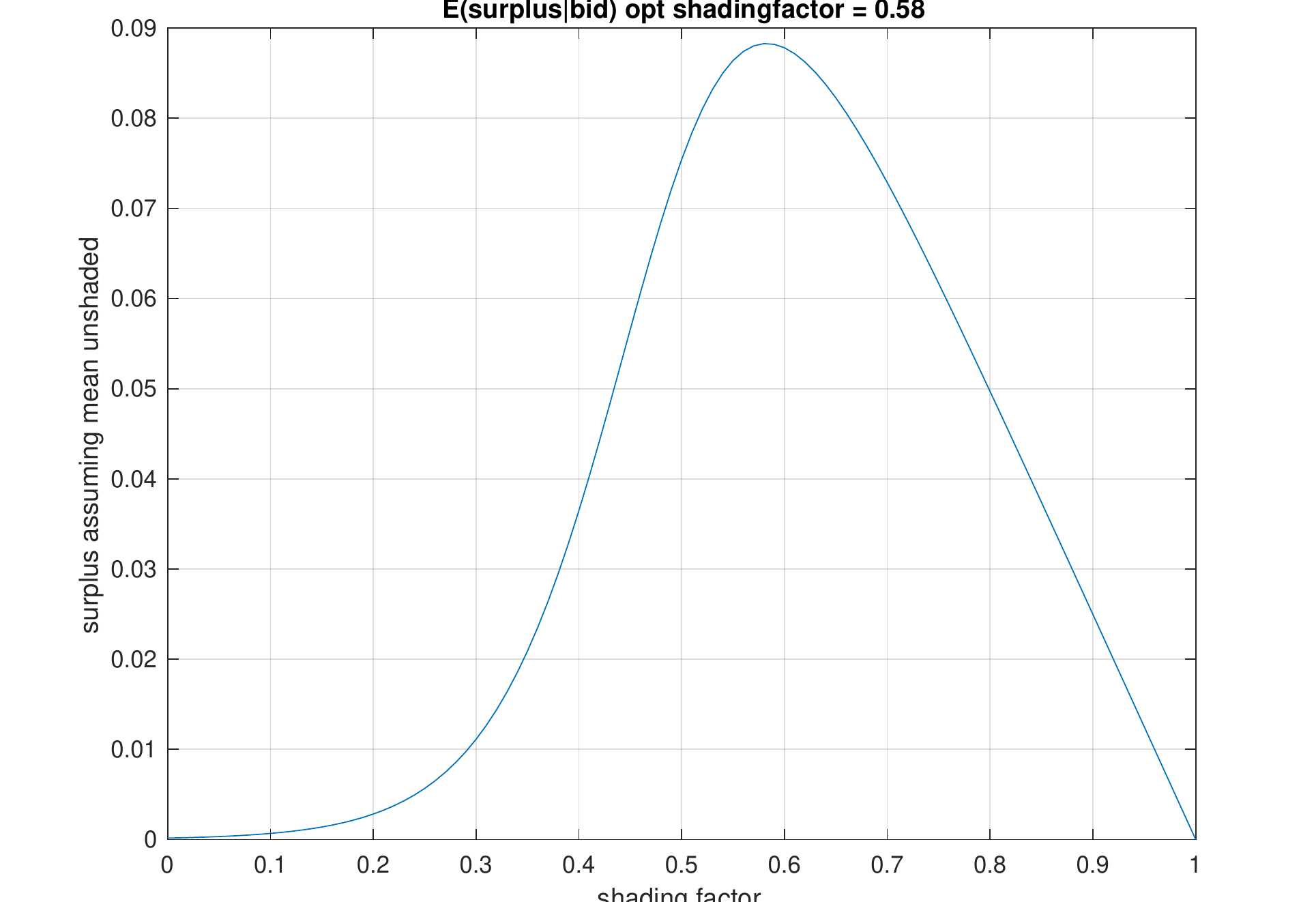}}
\caption{Top: Actual PDF for $\hat b$ (left) versus estimate (right); middle: CDF actual versus estimate; bottom: Surplus distribution actual versus estimate.}
\label{fig:mb2w_distributions}
\end{figure}

However, in practice, we rarely see such simple form of distributions.  Figure~\ref{fig:mb2w_distributions} shows an example of the empirical PDF of $\hat b$, including the derived surplus distribution.

Our approach breaks into two steps:
\begin{description}
    \item[Step 1] Estimate the distribution $\mathcal{D}_{\hat b\mid x_1, x_2, \ldots, x_k}$;
    \item[Step 2] Solve the maximization problem~\eqref{eq:max-surplus}.
\end{description}

\subsection{Distribution Estimation}

\label{subsubsec:distribution}
Given publisher and user attributions $x_1, \ldots, x_k$ and bid price $b$, we first train a classification model with historical data:
\begin{align}
    \text{Pr(win)}
    &= \text{cdf}_{\hat b}(b) =  F\left(w_0+\sum_{i=1}^k w_ix_i + \beta\, g(b)\right),
\end{align}
where $F$ is a fitting function that outputs a value between 0 and 1, which must be monotonically increasing in $b$ (higher bid price leads to higher winning rate), and $g(b)$ is a bid  transformation function such that $F\rightarrow 0$ as $b\rightarrow 0$, that is, as bid price goes to 0, the winning probability also goes to 0, and the weights to be learned are $w_0, w_1, \ldots, w_k$ and $\beta$.

For $g(b)$, we use the logarithm of bid price $g(b) =\log (b)$ so that
$g(b)\rightarrow -\infty$, as $b\rightarrow 0$. For $F$, we use the logistic function~\cite{cosma:linear, faraway:linear} so that $F(x)\rightarrow 0$ as $x\rightarrow -\infty$, with the constraint that $\beta > 0$.

Other forms $g(b)$ and $F$ can be explored, but our choices of simple forms, besides satisfying mentioned constraints, allow the maximization problem~\eqref{eq:max-surplus} in Step 2 to be solved efficiently. More details will follow later in Subsection~\ref{subsec:surplus-max}. With our choice of functions $F$ and $g(b)$, we have the following win-rate classification model:
\begin{align}
    \text{Pr(win)} &= \left(1+e^{-(w_0+\sum_{i=1}^k w_ix_i + \beta\, \log b)}\right)^{-1}, \label{eq:prod-win-rate}
\end{align}
which can be trained by gradient descent \cite{faraway:linear}. Note that the training should be constrained such that $\beta > 0$. In practice we found that it's not necessary since our learned $\beta$ without constraint turns out always positive.

\subsection{Surplus Maximization}
\label{subsec:surplus-max}
With a trained win-rate model from~\eqref{eq:prod-win-rate}, the optimal bid price $b^*$ can now be found by solving the optimization~\eqref{eq:max-surplus}:
\begin{align}
    b^* & = \mathop{\arg\max}_{b>0}(V-b)\, \text{logistic}\left(w_0+\sum_{i=1}^k w_ix_i + \beta\, \log b\right) \nonumber\\
    & = \mathop{\arg\max}_{b>0} (V-b)\left(1+e^{-w_0-\sum_{i=1}^k w_ix_i-\beta\log b}\right)^{-1} \nonumber\\
    & = \mathop{\arg\max}_{b>0} \frac{V-b}{1+e^{-\alpha} b^{-\beta}},
    \label{eq:func-b}
\end{align}
where $\alpha = w_0+\sum_{i=1}^k w_ix_i$.

We show below that, for $b>0$, there is a single optimum bid price $b^*$ which can be bounded from above and below. These bounds make it possible to implement a fast bisection search.

\begin{theorem}\label{thm:interval}
For any $\beta>0$,
\[
f(b)=\frac{V-b}{1+e^{-\alpha} b^{-\beta}}
\]
is maximized at some unique $b^*$ such that
\[  
\frac{\beta}{\beta+1+e^{\alpha} V^{\beta}}V \leq b^* < \frac{\beta}{\beta+1}V.
\]
\end{theorem}
\begin{proof}
Taking the derivative, we have
\[
f'(b) = \frac{\beta V - (\beta+1)b - e^\alpha b^{\beta+1}}{(1+e^{-\alpha}b^{-\beta})^2e^\alpha b^{\beta+1}}.
\]
Note that the denominator is always positive. Thus to find $b^*$ that maximizes $f(b)$ it's sufficient to consider the sign of
\[
h(b) \stackrel{\rm def}{=} \beta V - (\beta+1)b - e^\alpha b^{\beta+1}.
\]
Since $h(b)$ is a decreasing function in $b$, for any $b\in(0, V]$, $h(b)$ can be bounded as
\[
\beta V - (\beta+1)b - e^\alpha V^\beta\,b
\leq h(b) <
\beta V - (\beta+1)b,
\]
which implies that
\[
h\left(\frac{\beta}{\beta+1+e^{\alpha} V^{\beta}}V\right)\geq 0
\text{ and }
h\left(\frac{\beta}{\beta+1}V\right) < 0.
\]
Therefore there is a unique value $b^*\in \left[\frac{\beta}{\beta+1+e^{\alpha} V^{\beta}}V, \frac{\beta}{\beta+1}V\right)$ such that $h(b^*)=0$, and hence $f'(b^*)=0$. In other words, $f(b)$ is maximized at $b=b^*$.
\end{proof}

Theorem~\ref{thm:interval} allows us to implement a fast bisection search algorithm~\ref{alg:bisection} for the optimal bid price.
\begin{algorithm}
\caption{Bisection Algorithm Surplus Maximization}
\label{alg:bisection}
\begin{algorithmic}[1]
\Require \\
\begin{itemize}
    \item Model weights: $w_0, w_1, \ldots, w_k, \beta$;
    \item Feature values $x_1, x_2, \ldots, x_k$;
    \item $V$: expected value of the current ad opportunity
    \item $\epsilon>0$: minimum valid interval length
    \item $N$: maximum number of search steps
\end{itemize}
\Ensure $\beta > 0, V>0$
\Statex
\State $\alpha\leftarrow w_0+\sum_{i=1}^kw_ix_i$.
\State $b_{\min}\leftarrow \frac{\beta}{\beta+1+e^{\alpha} V^{\beta}}V$
\State $b_{\max}\leftarrow \frac{\beta}{\beta+1}V$
\For{$i = 1, 2, \ldots, N$}
    \State $\text{fp}_{\min} \leftarrow \beta V - (\beta+1)b_{\min} - e^\alpha b_{\min}^{\beta+1}$
    \State $\text{fp}_{\max} \leftarrow \beta V - (\beta+1)b_{\max} - e^\alpha b_{\max}^{\beta+1}$
    \State $r\leftarrow -\text{fp}_{\min}/(\text{fp}_{\max}-\text{fp}_{\min})$
    \label{step:ratio}
    \State $b\leftarrow (1-r)b_{\min} + r\,b_{\max}$
    \label{step:cut}
    \State $\text{fpb}\leftarrow \beta V - (\beta+1)b - e^\alpha b^{\beta+1}$
    \If{$\text{fpb} < 0$}
        \State $b_{\min}\leftarrow b$
    \Else
        \State $b_{\max}\leftarrow b$
    \EndIf
    \If{$b_{\max}-b_{\min} < \epsilon$}
        \State \textbf{break}
    \EndIf
\EndFor
\Statex
\Return $b$
\end{algorithmic}
\end{algorithm}
Starting with the minimum and maximum bounds on the surplus optimum, $b_{\min}= \frac{\beta}{\beta+1+e^{\alpha} V^{\beta}}V$
and $b_{\max}= \frac{\beta}{\beta+1}V$, 
per Theorem~\ref{thm:interval}, we know that the lower bound for optimum has positive derivative, and the high bound has negative. Bisection can divide the range and find the zero point for the derivative in at most $\log_2 [(b_{\max} - b_{\min})/\epsilon]$ steps; this logarithmic time is extremely desirable since the maximization search must run in real-time in the ad-server.

We found in practice that we could use the gradient information to speed up the search further. Rather than cutting the range in half each time ($r=0.5$; step~\ref{step:ratio}), after testing the gradient of the minimum and maximum bid points, we use our knowledge that the surplus function is convex and so derivatives shorten close to the optimum. We calculate the ratio between the surplus derivative at min and max bid locations, and then use that estimate for the relative distance to the optimum in bid space. Steps~\ref{step:ratio} and~\ref{step:cut} of the pseudo-code show this modification to $r$. Empirically we observed that the bisection ends in less 10 iterations to achieve a sufficient precision, which is controlled by $\epsilon$.

\section{Implementation}
\label{sec:implementation}
The features used for predicting win probability comprise 12 variables extracted from the HTTP of an incoming bid request, along with \texttt{log(bid price)} and \texttt{log(bid price before shading)}. The HTTP attributes include the requesting page (e.g., \texttt{yahoo.com}), device type (e.g., \texttt{desktop}), hour of day; day of week,  country, user segment, and other variables.
All of the HTTP features are encoded to be binary variables.

For production we use one week of historical data for training so that weekly patterns are captured. The training data typically contain over a billion of bid requests with less than 100K encoded features. For the curve fit, we used the LogisticRegression method that is part of the PySpark pyspark.ml.classification library~\cite{pyspark}, which is distributed. The training time depends on the number of allocated machines. With less than 100 machines the training can finish within a few hours.

At run-time, the shading algorithm needs to respond to millions of requests per second peak load, within 100 milliseconds for all systems.
In order to meet these speed constraints, bid shading has to minimize the number of computations that it performs. 
In terms of memory, by using a single global model, memory consumption is kept to less than $100K$ floating point numbers. 
In terms of time,  shading optimization averages just less than 20 floating-point operations per request.

\section{Shading Insights}
\label{sec:insights}
Here we describe a few features that we observed to be predictive in the win-rate model. 
The numerical features logarithm of bid price before shading and logarithm of bid price are both highly predictive\footnote{In the following, the regression coefficient is labeled $w$ and $PR$ is the observed positive rate of the binary variable} ($w$  = -0.39 and 0.565; McFadden $R^2$=0.24 and 0.20 respectively \cite{freese, ucla}). The high predictiveness of bid price before shading - and yet negative sign when included with bid price - is consistent with previous observations that bid shading tends to be deeper in auctions with higher valuations  \cite{chakravorti1995auctioning, battigalli2003rationalizable, hortaccsu2018bid}.

The top binary feature in terms of impact on win probability is \texttt{is\_new\_user} ($w=0.831$; $PR$=0.52), which is associated with an increase in chance of winning the auction (since bid prices are lower).
auctions. 
\texttt{hour\_of\_day=6am} (user local hour) ($w=-0.267$; $PR$=0.01) is associated with a drop in the probability of winning, 
likely due to the reduction in supply \cite{kittslookahead}. \texttt{country=US} ($w=-0.110$; $PR$=0.84) decreases the chance of winning; and the largest 768x1024 ads also are less likely to be won ($w=-0.267$; $PR$=0.01).

The predictability of time, user, and other features, for estimating auction clearing prices, suggest that shading should be effective, as noted in work~\cite{pownall2013bidding} on the preconditions for shading described in Section~\ref{sec:related-work}.

\section{Comparison to Benchmarks}
\label{sec:experiments}
We ran several of the algorithms in Section \ref{sec:related-work} as benchmarks. These included: (1) Sell-Side Shading Service (S4) \cite{adx:rubicon, rubicon:EMR, appnexus:BPO, google:rtb, google:rtbprotocol}, (2) Non-linear Segment-Based (SEG)  \cite{niklas:shading},  Distribution estimator with Normal (NRML), Exponential (EXP) Distributions \cite{wu2015predicting}, Logistic Regression (LR) \cite{dsp:lr} and Unshaded (Uns). The win-rate based algorithm in this paper is labeled \emph{WR} in the tables to follow.

The prior work benchmarks aren't ideal - the win distribution approaches \cite{wu2015predicting} don't explicitly maximize surplus and so we expect them to not perform as well. The S4 algorithm seems to be geared towards maintaining win rate. Nevertheless, we have included them not only to compare to prior work but also to quantify the gain that surplus maximization approaches can deliver in practice. 

Unlike the other benchmarks, the SEG algorithm does maximize surplus \cite{niklas:shading}. Under a favorable selection of segments, the Segment-Based algorithm might even be tuned to perform as well or better than the current method, despite the scaling problem with using more features. Our purpose in showing these benchmarks isn't to claim that this particular algorithm outperforms the others in all metrics, but rather to show that surplus maximizers have an advantage, to quantify the gain, and to note that WR, which is fully automated, uses all available features to estimate the surplus landscape, and has excellent memory and speed properties, performs comparable to other reported approaches.

The experiments below (except ones with S4) were run on auctions for which the minimum bid prices to win were known. Using this data it was possible to calculate surplus performance as a percentage of the total optimal surplus:
\begin{align*}
    \text{\% opt surplus} &\stackrel{\rm def}{=}
    \frac{\sum_i(V_i-b_i)\mathbf{I}(b_i>\hat b_i)}{\sum_i(V_i-b_i)},
\end{align*}
i.e., the surplus achieved by a particular algorithm out of total available surplus by bidding optimally.
Spend and impression performances can be defined similarly.

The algorithms were tested on one day of auction data. For fair comparison training is done on data from the previous day, since not all algorithms are designed to be trained on multiple days of data. All the bid requests are scored by each algorithm, so all algorithms operate on the same set of records. The results are shown in Table \ref{tab:Exp1}.

The distribution estimator methods (NRML, EXP) estimate the minimum bid to win and so are not expected to do well in maximizing surplus. As a group they were about 7\% below WR.
The Nonlinear Segment method generated the second highest surplus besides WR. This makes sense given that it is a legitimate surplus maximizer. WR generates the highest surplus (50.6\%). In sum, the surplus maximizers produced the most surplus, which was expected.

\newcolumntype{H}{>{\setbox0=\hbox\bgroup}c<{\egroup}@{}}

\begin{table}
\setlength{\tabcolsep}{3pt}
  \begin{tabular}{ cHccccHcc }
    \toprule
    Metric&Opt&WR&SEG&NRML&LR&Gam&EXP&Uns\\
    \midrule
\%opt surplus&100\%&50.6\%&49.0\%&48.0\%&47.3\%&46.7\%&46.0\%&0\%\\
\%opt spend&100\%&41.7\%&56.0\%&42.7\%&39.8\%&32.5\%&31.1\%&176\%\\
\%opt imps&100\%&56.6\%&49.1\%&53.1\%&50.3\%&43.4\%&42.6\%&100\%\\
avg shading factor&0.55&0.6&0.55&0.62&0.61&0.44&0.42&1.00\\
  \bottomrule
\end{tabular}
  \caption{Benchmark Algorithms}
  \label{tab:Exp1}
\end{table}

We also compared an S4 algorithm from an anonymous SSP. We had to separate this analysis due to a service issue. When using the S4 for real-time bidding, the service disabled the minimum bid to win functionality. As a result, we were unable to do an optimality analysis. 

Overall, the S4 delivered about 15\% more impressions than WR - as noted SSPs have an incentive to try to monetize more traffic. However it delivered about 4.3\% lower surplus. The bidding distribution from the S4 is shown in 
Figure \ref{fig:winrate_threedistributions}.
\begin{figure}
\centering
\subfloat{\includegraphics[trim=0cm 1cm 0cm 1cm,clip, width=1\columnwidth]{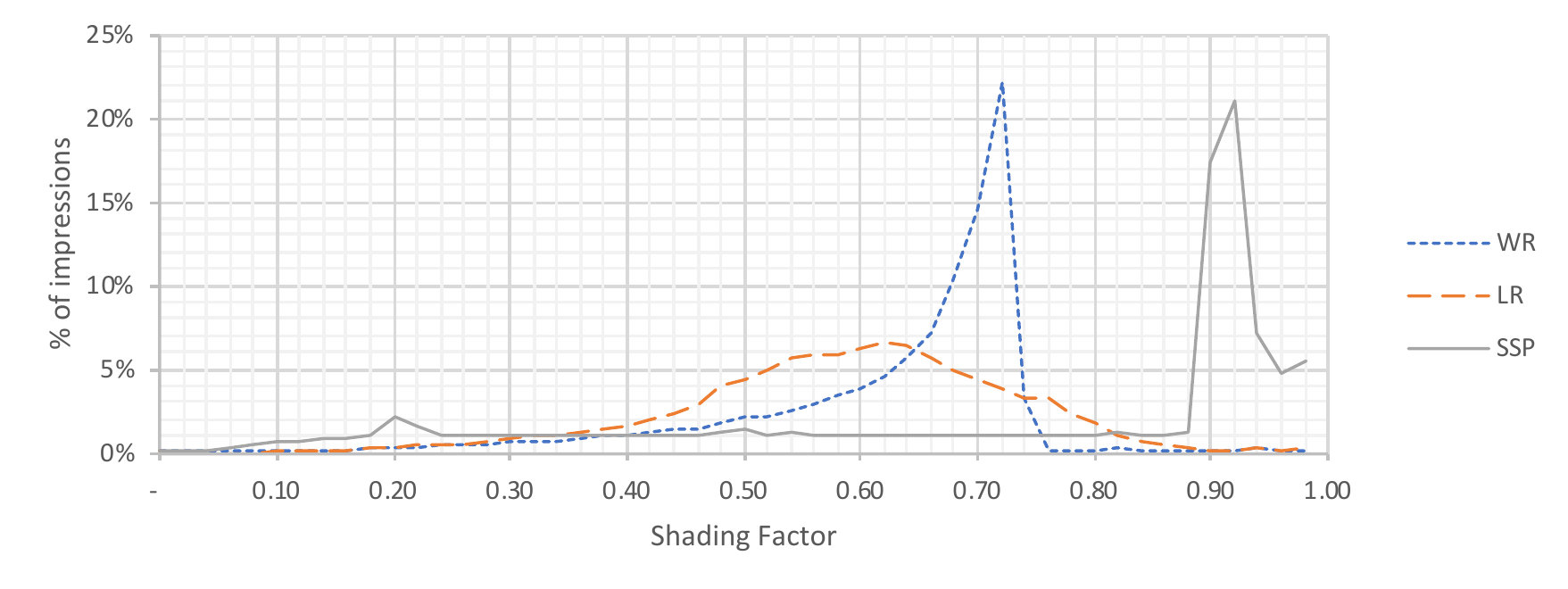}}
\caption{Shading factor distributions for three algorithms. S4 has more shallow shading factors.}
\label{fig:winrate_threedistributions}
\end{figure}
Whereas the SSP's shading distribution is right-skewed, with most shading at 90\% and above, the WR distribution - which generates more surplus - is left-skewed, with most shades below 72\%. It seems likely that the S4 is geared towards generating high sales, but not necessarily high advertiser surplus.

\section{Production Results}
\label{sec:Production}
After rolling out the WR algorithm, we were able to monitor its online performance by maintaining a percentage of traffic that was randomly allocated to each algorithm. 
The analysis shown in Table~\ref{tab:Exp2} spans about two months, during which time all algorithms were automatically updated at daily basis.
\begin{table}
  \begin{tabular}{ ccccc }
    \toprule
Metric&WR&LR&SEG&Uns\\
    \midrule 
\% opt surplus&46.7\%&44.8\%&38.2\%&0.0\%\\
\% opt spend&79.1\%&72.6\%&89.9\%&410\%\\
\% opt imps&60.3\%&51.4\%&56.0\%&100\%\\
avg shading factor& 0.55& 0.53 & 0.59 & 1.00 \\
  \bottomrule
\end{tabular}
\caption{Production Results}
  \label{tab:Exp2}
\end{table}
Overall WR captured 46.7\% of the maximum possible surplus, whereas Non-linear captured 38\%. Bid prices on WR were about 45\% lower than their unshaded prices.

Note that in a real-time bidding system, campaigns usually have finite budgets, and budget controller is a necessary component in such a system. The production performance of a bid shading algorithm relies on how well it works together with the budget controller. In a simplified view, a reasonable controller is expected to spend all the daily budget, and hence the budget \emph{saved} by a bid shading algorithm, that is, the surplus, would be spent again to buy more impressions, thus leads to lower eCPM, eCPC, and eCPA~\cite{dsp:metrics}. Indeed, as shown in Table~\ref{tbl:prod-metrics}, with similar spend WR achieved significant improvements on these business metrics.
\begin{table}
  \begin{tabular}{ lccccc }
    \toprule
    A/B Testing & Spend & Surplus & eCPM & eCPC & eCPA \\
        \midrule
    WR v.s. LR & +1.3\% & +1.4\% & -7.4\% & -4.5\% & -2.7\% \\
    WR v.s. SEG & +1.2\% & +2.5\% & -5.4\% & -5.5\% & -3.9\% \\
        \bottomrule
\end{tabular}
\caption{Improvements on Business Metrics}
\label{tbl:prod-metrics}
\vspace*{-5ex}
\end{table}

\section{Conclusion}
\label{sec:conclusions}
There is evidence that First Price Auctions have created problems for advertisers. Average traffic prices are higher, with  estimates ranging between 5\% and 50\% 
\cite{rubicon:EMR, kitts:fpa, appnexus:BPO, hearts2018:sold}.
\cite{kitts:fpa} also reported that after their SSP switched to First Price, 10\% of advertisers actually discontinued bidding. 
Our experiments confirm these findings; without a shading solution, CPM would approximately double. 

DSPs are required to compute the private value of impressions based on advertiser  parameters, and they also execute a large number of trades, and so can build up an ability to predict auction prices. This makes it possible to implement rational  shading similar to other industries \cite{lengwiler2010auctions, hendricks1989collusion, hortaccsu2018bid}. 
Advertiser bids follow the value of traffic, and this follows daily, hourly, and site patterns. 
As a result, auction prices will always have structure that can be used by some advertisers with other advertisers have less flexibility.

The surplus maximization approach of this paper delivered about 7\% higher surplus than naive methods just designed to submit the probable clearing price. Furthermore, when integrated with budget controllers, it significantly reduced eCPM, eCPC and eCPA by 3\%~--~7\%. Publicly available data shows medium sized DSPs managing between 260 to 1 billion US dollars in advertiser spend \cite{dsprev:idc}. The Shading gains reported in this paper therefore represent 18 to 100 million US dollars in additional yield that is provided to advertisers.
Shading has an enormous impact on advertiser profitability. Now that the online ad industry has increasingly shifted to First Price Auctions, it seems likely that the new advertising technology arms race will be in the domain of bid shading.



\bibliographystyle{ACM-Reference-Format}

\balance
\bibliography{references}


\end{document}